\documentclass{article}
\usepackage{amsmath, amssymb, amsfonts, amsthm}
\usepackage[T1]{fontenc}
\usepackage{fullpage}
\usepackage{thmtools}
\usepackage{thm-restate}
\usepackage{algorithmicx}
\usepackage{algorithm}
\usepackage{algpseudocode}
\usepackage{graphicx}
\usepackage[normalem]{ulem}
\usepackage{cancel}
\usepackage[square,numbers]{natbib}
\usepackage{textcomp}
\usepackage{enumitem}
\usepackage{booktabs}
\usepackage{tikz}
\usepackage{bm}
\usepackage{standalone}
\usepackage{accents}
\usepackage{lineno}
\usepackage{graphicx}
\usepackage[]{mdframed}
\usepackage{thmtools}
\makeatletter
\renewcommand\thmt@autorefsetup{%
	\@xa\def\csname\thmt@envname autorefname\@xa\endcsname\@xa{\thmt@thmname}%
	\ifthmt@hassibling
	\@xa\def\@xa\thmt@autoreffix\@xa{%
		\@xa\global\@xa\let\csname the\thmt@envname\@xa\endcsname
		\csname the\thmt@sibling\endcsname
		\def\thmt@autoreffix{}%
	}%
	\fi
}
\makeatother

\makeatletter
\theoremstyle{plain}
\newtheorem{theorem}{Theorem}
\newtheorem{lemma}{Lemma}

\theoremstyle{definition}
\newtheorem*{example}{Example}

\newtheorem*{maindefinition}{Main definition}

\newtheorem*{2-period-queries}{2-Period Queries}
\newtheorem*{period-queries}{\boldmath$(1+\epsilon)$-Period Queries}
\newtheorem*{substring-compression-queries}{Substring Compression Queries}
\newtheorem*{longest-substring-repeat}{Longest Substring Repeat}
\makeatother

\makeatletter
\DeclareRobustCommand{\cev}[1]{%
	{\mathpalette\do@cev{#1}}%
}
\newcommand{\do@cev}[2]{%
	\vbox{\offinterlineskip
		\sbox\z@{$\m@th#1 x$}%
		\ialign{##\cr
			\hidewidth\reflectbox{$\m@th#1\vec{}\mkern4mu$}\hidewidth\cr
			\noalign{\kern-\ht\z@}
			$\m@th#1#2$\cr
		}%
	}%
}
\makeatother

\newcommand\dd{\,..\,}
\DeclareMathOperator{\E}{\mathbb{E}}
\DeclareMathOperator{\size}{\mathsf{size}}

\title{Closed Repeats}

\author{
	\textsc{Dmitry Kosolobov}\\
	\normalsize Ural Federal University, Ekaterinburg, Russia\\
	\normalsize \texttt{dkosolobov@mail.ru}
}

\date{}

\begin{document}

\maketitle

\begin{abstract}
	Much research in stringology focuses on structures that can, in a way, ``grasp'' repeats (substrings that occur multiple times) as, for example, the so-called runs, a.k.a.\ maximal repetitions, compactly describe all tandem repeats. In this paper we introduce closed repeats: given a string $s$, its non-empty substring $s[i\dd j]$ is a right (left) closed repeat if its closest occurrence $s[i'\dd j']$ with $i' > i$ cannot be ``extended'' to the right (respectively, left) matching $s[j{+}1] = s[j'{+}1]$ (respectively,  $s[i{-}1] = s[i'{-}1]$); the repeat is closed if it is both left and right closed. We note that the closed repeats correspond to the maximal closed substrings recently proposed by Badkobeh, De Luca, Fici, and Puglisi and they include all runs as a special case. We prove that the number of right/left closed repeats is $O(n \log n)$, where $n$ is the length of $s$, and we show that this bound is tight. The (right/left) closed repeats can be computed in the optimal time $O(n\log n)$; as we prove, the time cannot be lower than $\Omega(n\log\sigma)$ over a general size-$\sigma$ alphabet even if there are only $O(n)$ closed repeats. As an application, we describe data structures using the closed repeats for substring queries: finding the period of the substring provided it is ``periodic'', finding the longest repeat in the substring, computing the rightmost LZ77 parsing of the substring.


\noindent\textbf{Keywords:} closed repeats, maximal repeats, repetitions, runs, closed words, LZ77, Lempel--Ziv
\end{abstract}

\algtext*{EndIf}
\algtext*{EndWhile}
\algtext*{EndFor}

\section{Introduction}

The fields of stringology and combinatorics on words have a long history of studying structures on strings that help to ``grasp'' repeats (substrings that occur at least twice). As a primary example, much research has been devoted to tandem repeats (also called squares), which can be ``grasped'' by the so-called runs, a.k.a.\ maximal repetitions~\cite{KolpakovKucherov,MainLorentz}. Multiple other concepts were proposed to tackle problems related to non-tandem repeats, like maximal repeats~\cite{BlumerEtAl,Crochemore,Gusfield,Raffinot}, gapped repeats~\cite{BLPS,KolpakovKucherov2,KPPK}, subrepetitions~\cite{Kolpakov,KPPK}, closed substrings~\cite{BDFP,BDFP2,Fici,ParshinaPuzynina}, maximal exponent repeats~\cite{BadkobehCrochemoreToopsuwan}, and a few others.

In this paper we propose a new combinatorial object: given a string $s$ of length $n$, its substring $s[i\dd j]$ is called a \emph{right (left) closed repeat} if its closest occurrence $s[i'\dd j']$ with $i' > i$ cannot be ``extended'' to the right (respectively, left) matching $s[j{+}1] = s[j'{+}1]$ (respectively, $s[i{-}1] = s[i'{-}1]$); the repeat is \emph{closed} if it is both left and right closed (a precise definition will be given below). 
It would be more accurate to say that the closed repeats are not a novel concept \emph{per se} but rather a different view on the recently introduced maximal closed substrings by Badkobeh, De Luca, Fici, and Puglisi~\cite{BDFP2,BDFP,BadkobehFiciLiptak} since there is a one-to-one correspondence between the closed repeats and the non-trivial maximal closed substrings. We believe, however, that certain applications of these objects and questions about them come more naturally if they are viewed as repeats. Some examples of repeat-related problems and such combinatorial questions are provided in the sequel and the right/left closed repeats is another variation on this theme, which stems from this view.
%

As in~\cite{BDFP2,BDFP}, we point out that the runs are a particular case of the closed repeats. However, in contrast to the runs~\cite{BannaiEtAl,KolpakovKucherov}, the number of closed repeats is not $O(n)$: we prove that there are at most $O(n\log n)$ right/left closed repeats\footnote{All logarithms in the paper are in base $2$.} (the same bound for closed repeats was essentially established in~\cite{BDFP2}) and this bound is tight since the random binary string has $\Omega(n \log n)$ closed repeats in expectation, which solves an open problem posed in~\cite{BDFP2,BDFP}. As a further distinction, whereas the runs can be calculated in $O(n)$ time on general ordered alphabets, we show that the closed repeats cannot be computed faster than in $\Omega(n\log\sigma)$ time on a general ordered alphabet of size $\sigma \ge n^{\Omega(1)}$ even when the number of closed repeats is $O(n)$.

Despite these somewhat ``negative'' results, the closed repeats have many advantages as a tool that ``nails down'' the repeats to their occurrences: as it is easy to observe, each occurrence of a repeat, except the rightmost one, can be ``extended'' to the right/left to a (right/left) closed repeat. It is akin to how the occurrences of tandem repeats can be extended to runs. We also note that the (right/left) closed repeats can be computed in the optimal time $O(n\log n)$ by a straightforward adaptation of the algorithm from \cite{BDFP2}, which originally was devised to find all maximal closed substrings (see details in Section~\ref{sec:computation}). As a demonstration for the utility of the closed repeats, we complement the paper with applications, not all of which are strictly novel but they serve as instructive examples, in our opinion. 

We propose three data structures that support certain queries on arbitrary substrings of a given string $s$ of length $n$. All our constructions take $O(n\log n)$ space (mainly to store the closed repeats). First, we can answer in $O(1)$ time whether the query substring is periodic, i.e., has period $p$ that fits in the substring twice, and we find the minimal period $p$ if so. This data structure is inferior to the state of the art~\cite{KRRW} in terms of space but the ease with which we achieve this query time might be of particular interest here. Second, for any constant $\epsilon > 0$ fixed in advance, we can compute in $O(\log^\epsilon n)$ time the longest repeat in the query substring (i.e., the longest string that occurs in the query substring at least twice). Apparently, this type of queries was not studied before. Again, the solution is quite simple but relies on a known complicated data structure for orthogonal range minima queries~\cite{Nekrich}. Third, we can compute the so-called rightmost LZ77 parsing (the definition follows) of the query substring in $O(z\log\log\ell)$ time, where $z$ is the number of phrases in the parsing and $\ell$ is the length of the substring. The state of the art with this space~\cite{KKFL,KRRW,Zhou} is slightly slower: it takes $O(z\log\log n)$ time; but it can also achieve this time with $O(n\log\log n)$ space. However, this known solution, as all other works on this problem, computes an LZ77 parsing that is not necessarily rightmost. The rightmost parsing is known to be noticeably superior for compression~\cite{FerraginaNittoVenturini,KosolobovBitOpt,KVNP}. Apparently, ours is the first solution that computes the rightmost parsing. We note again that our solution is simple. A slightly more detailed context on these three considered problems is provided in Section~\ref{sec:applications}.

The paper is structured as follows. Section~\ref{sec:closed-repeats} defines the (right/left) closed repeats and relates them to some other well-known repetitive structures in strings. In Section~\ref{sec:combinatorics} we prove that there are at most $O(n\log n)$ (right/left) closed repeats and we investigate the related combinatorics. In Section~\ref{sec:computation} we obtain the lower bound $\Omega(n\log\sigma)$ for the computation of all closed repeats over a general ordered alphabet of size $\sigma  \ge n^{\Omega(1)}$; also a reference is provided to a known practical algorithm from~\cite{BDFP2} that can be adapted to find all (right/left) closed repeats in the optimal time $O(n\log n)$. Section~\ref{sec:applications} discusses applications. We conclude with open problems in Section~\ref{sec:conclusion}.

\section{Closed Repeats and Other Repetitive Structures}
\label{sec:closed-repeats}

For a string $s = c_1 c_2 \cdots c_{n}$, denote by $|s|$ its length $n$. We write $s[i]$ for the letter $c_i$ and $s[i \dd j]$ for the \emph{substring} $c_i c_{i + 1} \cdots c_j$, assuming $s[i \dd j]$ is empty if $i > j$. The \emph{empty string} is denoted by $\epsilon$. We also write $s[i\dd j)$ for $s[i\dd j{-}1]$. We say that a string $t$ \emph{occurs} in $s$ at position $i$ if $s[i\dd i{+}|t|) = t$. A substring $s[i\dd n]$ is called a \emph{suffix} of $s$ and a substring $s[1\dd i]$ is called a \emph{prefix}. Denote $[i\dd j] = \{i, i{+}1, \ldots, j\}$. A string $t$ is called a \emph{repeat} in a string $s$ if $t$ occurs in $s$ at least twice.

\begin{maindefinition}
	Given a string $s$ of length $n$, its non-empty substring $s[i\dd j]$ is called a \emph{right} (respectively, \emph{left}) \emph{closed repeat} if it has an occurrence $s[i'\dd j'] = s[i\dd j]$ with $i' > i$ such that  $s[i\dd j]$ does not occur at positions $i{+}1, i{+}2, \ldots, i'{-}1$ and either $s[j{+}1] \ne s[j'{+}1]$ or $j' = n$ (respectively, $s[i{-}1] \ne s[i'{-}1]$ or $i = 1$); the repeat $s[i\dd j]$ is \emph{closed} if it is both left and right closed. We call $s[i'\dd j']$ the \emph{next occurrence of $s[i\dd j]$}.
\end{maindefinition}

To put it differently, $s[i\dd j]$ is a closed repeat if its next occurrence cannot be ``extended'' to the left and to the right matching $s[i{-}1]$ (if $i > 1$) and $s[j{+}1]$ (if $j < n$), respectively. It is important to emphasize that we treat as different those closed repeats that are equal as strings but occur at different positions in $s$. 

\begin{example}
	The string $s = banana$ has the only closed repeat $s[2\dd 4]$, three right closed repeats $s[2\dd 4]$, $s[3\dd 4]$, $s[4\dd 4]$, and three left closed repeats $s[2\dd 2]$, $s[2\dd 3]$, $s[2\dd 4]$. 
\end{example}

We consider the (right/left) closed repeats as position-dependent structures that can ``describe'' all repeats in the following sense: given an arbitrary repeat $t$, each its occurrence $s[i\dd j] = t$, except the rightmost one, is contained in at least one right closed repeat of the form $s[i\dd j{+}r]$, for some $r \ge 0$. To see this, one can take the next occurrence $s[i'\dd j']$ of $s[i\dd j]$ and take as $r$ the length of the longest common prefix of $s[j{+}1\dd n]$ and $s[j'{+}1\dd n]$. Note that the right closed repeat containing $s[i\dd j]$ is not unique: for instance, the repeat $s[1\dd 1]$ in $s = abcababc$  is contained in two right closed repeats $s[1\dd 2] = ab$ and $s[1\dd 3] = abc$. Analogously, one can show that $s[i\dd j]$ is contained in a (not necessarily unique) left closed repeat $s[i{-}\ell\dd j]$ and a (not necessarily unique) closed repeat $s[i{-}\ell\dd j{+}r]$, for some $\ell \ge 0$ and $r \ge 0$.

Let us relate the closed repeats to some other well-known combinatorial objects describing repeats.

\textbf{Maximal repeats.}
A string $t$ is called a \emph{maximal repeat} of a string $s$ if $t$ occurs in $s$ at least twice and, for each letter $a$, the strings $at$ and $ta$ have strictly less occurrences in $s$. Note that the maximal repeats, unlike the closed repeats, are not tied to specific positions in $s$. Clearly, every closed repeat is a maximal repeat (when treated as a separate string, not a substring). Further, every maximal repeat $t$ has at least one occurrence in $s$ that is a right (respectively, left) closed repeat; but $t$ does not necessarily occur as a closed repeat, as the following example shows: the string $bacacab$ has a maximal repeat $a$ but all occurrences of $a$ in $s$ are not closed repeats.

\textbf{Runs.}
A \emph{period} of a string $t$ is an integer $p > 0$ such that $t[i] = t[i{+}p]$, for all $i \in [1\dd |t|{-}p]$. A \emph{run} or \emph{maximal repetition} in a string $s$ is a substring $s[i\dd j]$ whose minimal period $p$ is such that $2p \le j - i + 1$ and neither $s[i{-}1\dd j]$ (if $i > 1$) nor $s[i\dd j{+}1]$ (if $j < n$) have $p$ as their period. The runs can be considered as a structure that ``describes'' all \emph{tandem repeat}, i.e., substrings $s[i\dd i{+}p)$ such that $s[i\dd i{+}p) = s[i{+}p\dd i{+}2p)$: one can maximally ``extend'' $s[i\dd i{+}2p)$ left and right with the period $p$, thus obtaining the unique run ``containing'' $s[i\dd i{+}2p)$. It is known that there are at most $n$ runs in $s$ (see~\cite{BannaiEtAl,KolpakovKucherov}). The basic combinatorics of periods implies that, for a run $s[i\dd j]$ with the minimal period $p$, the substring $s[i\dd j{-}p]$ is a closed repeat and its next occurrence is $s[i{+}p\dd j]$. Due to this relation, the runs form a particular case of closed repeats.

\textbf{Maximal closed substrings.}
A \emph{border} of a string $t$ is a string that is both prefix and suffix of $t$. The border is \emph{proper} if it is not equal to $t$. In~\cite{BDFP} the following definition was proposed: a substring $s[i\dd j]$ is called a \emph{maximal closed substring} of the string $s$ if it has a non-empty proper border that only occurs in $s[i\dd j]$ as a suffix and prefix and neither the substring $s[i{-}1\dd j]$ (if $i > 1$) nor $s[i\dd j{+}1]$ (if $j < n$) satisfy the same property (i.e., they do not have such non-empty borders that only occur as prefix and suffix).\footnote{In~\cite{BDFP} substrings of length one also were called maximal closed substrings; we excluded this special case from the definition so that there is a complete correspondence to the closed repeats.} There is a one-to-one correspondence between the closed repeats and the maximal closed substrings: if $s[i\dd j]$ is a closed repeat and $s[i'\dd j']$ is its next occurrence, then $s[i\dd j']$ is a maximal closed substring and $s[i\dd j]$ is its longest border; conversely, if $s[p\dd q]$ is a maximal closed substring whose longest border is $s[p\dd p{+}b]$, then $s[p\dd p{+}b]$ is a closed repeat and $s[q{-}b\dd q]$ is its next occurrence. According to this correspondence, the closed repeats and the maximal closed substrings are essentially the same objects.

\textbf{\boldmath Maximal $\alpha$-gapped repeats.} For $\alpha > 1$ and a string $s$ of length $n$, its substring $s[i\dd j]$ with a decomposition $s[i\dd j] = uvu$ is called a \emph{maximal $\alpha$-gapped repeat}~\cite{KPPK,KPPK2} if $|uv| \le \alpha|u|$ and $u$ cannot be ``extended'' to the right or left: either $i = 1$ or $s[i{-}1] \ne s[j{-}|u|]$, and either $j = n$ or $s[i{+}|u|] \ne s[j{+}1]$. (Note that the same substring $s[i\dd j]$ with different decompositions $uvu$ may form different maximal $\alpha$-gapped repeats.) In other words, it is a pair of substrings $u$ at a distance at most $\alpha|u|$ that cannot be ``extended'' to the right or left. It was proved in~\cite{CrochemoreKolpakovKucherov,IKoppl,GIIKM,GIIKM2} that there are at most $O(\alpha n)$ maximal $\alpha$-gapped repeats in $s$. Evidently, if $s[p\dd q]$ is a closed repeat and $s[p'\dd q']$ is its next occurrence, then $s[p\dd q']$ is a maximal $\alpha$-gapped repeat for $\alpha \ge (p' - p) / (q - p + 1)$. The converse relation, however, is not straightforward: $u$ may have occurrences in a substring $uvu = s[i\dd j]$ other than the prefix and suffix so that no closed repeat at position $i$ has next occurrence at position $j - |u| + 1$ (for example, consider the maximal $4$-gapped repeat $abaa$). The concepts of $\alpha$-gapped repeats and closed repeats seem to ``grasp'' different aspects of repeats: the $\alpha$-gapped repeats provide a detailed view on repeats with relatively close occurrences whereas the closed repeats may encode information about arbitrarily distant repeats but in a less detailed manner.

\section{Combinatorics of Closed Repeats}
\label{sec:combinatorics}

Perhaps, the most important fact about the (right/left) closed repeats is that there are at most $O(n\log n)$ of them, which makes them a relatively compact structure that stores the information about all repeats, in a sense. This was essentially proved in~\cite{BDFP2} for the closed repeats using an algorithm computing all maximal closed substrings (the bound followed since the algorithm worked in $O(n\log n)$ time). We provide a direct combinatorial proof adapted to the right/left closed repeats.

\begin{theorem}
	Any string of length $n$ contains at most $O(n\log n)$ left closed repeats and right closed repeats.\label{thm:upper}
\end{theorem}
\begin{proof}
	Consider a string $s$ of length $n$. Denote by $T$ the suffix tree of the string $s\$$, where $\$$ is a special letter that does not occur in $s$. We call an edge $xy$ of $T$ connecting two nodes $x$ and $y$ light if $\size(x) \ge 2\cdot\size(y)$, where $\size(x)$ denotes the number of leaves in the subtree rooted at $x$. Evidently, for any leaf $z$, the root--$z$ path (i.e., the path from the root to $z$) has at most $\log n$ light edges. Hence, there are at most $n\log n$ pairs $(z,xy)$ such that $z$ is a leaf and $xy$ is a light edge on the root--$z$ path.

	Consider a right closed repeat $s[i\dd j]$ and its next occurrence $s[i'\dd j']$. It is straightforward that the tree $T$ has a node $x$ such that the string $s[i\dd j]$ is written on the root--$x$ path and $x$ has children $y$ and $y'$ such that $s[j{+}1]$ is the first letter written on the edge $xy$ and $s[j'{+}1]$ is the first letter on $xy'$ (assuming $s[j'{+}1] = \$$ if $j' = n$). Observe that at least one of the edges $x y$ and $x y'$ must be light. Denote by $z$ and $z'$ the leaves of $T$ corresponding to the suffixes $s[i\dd n]$ and $s[i'\dd n]$, respectively. We associate the repeat $s[i\dd j]$ with the pair $(z,xy)$ if the edge $xy$ is light, and with $(z',xy')$ otherwise (in this case $xy'$ is light).
	
	Given any pair $(z,xy)$ such that $z$ is a leaf and $xy$ is a light edge on the root--$z$ path, we claim that the pair could be associated with at most two right closed repeats. Denote by $t$ the string written on the root--$x$ path. Let $s[k\dd n]$ be the suffix corresponding to $z$. The pair $(z,xy)$ could be associated either with the right closed repeat $s[k\dd k{+}|t|)$ (if any) or with a closed repeat $s[i\dd j] = t$ whose next occurrence is $s[k\dd k{+}|t|)$. There is at most one such $s[i\dd j]$ since $s[i\dd j]$ must be the closest to $s[k\dd k{+}|t|)$ occurrence of $t$ that precedes $s[k\dd k{+}|t|)$. Thus, since there are at most $n\log n$ such pairs $(z,xy)$, the number of right closed repeats is at most $2n\log n$.
	
	For left closed repeats, one can observe that whenever $s[i\dd j]$ is a left closed repeat and $s[i'\dd j']$ is its next occurrence, the substring of the reversed string $\cev{s} = s[n]\cdots s[2]s[1]$ that corresponds to $s[i'\dd j']$ is a right closed repeat in $\cev{s}$ and its next occurrence in $\cev{s}$ is the substring corresponding to $s[i\dd j]$. Hence, the bound on the number of left closed repeats in $s$ follows from the bound on the right closed repeats in $\cev{s}$.
\end{proof}

In particular, Theorem~\ref{thm:upper} implies that the number of closed repeats is $O(n\log n)$, which was already proved in~\cite{BDFP2} in terms of maximal closed substrings. This bound is tight due to the following result.

\begin{theorem}
	The uniformly random binary string of length $n$ has $\Omega(n\log n)$ closed repeats in expectation.\label{thm:lower}
\end{theorem}
\begin{proof}
Consider a uniformly random binary string $s$ of length $n$. For a position $i \in [2\dd n/2]$ and an integer $b \in [1\dd \frac{1}{3} \log n]$, denote by $I(i,b)$ an indicator random variable that is equal to $1$ iff $s[i\dd i{+}b)$ is a closed repeat in $s$. Then, the total number of closed repeats is at least $I = \sum_{i \in [2\dd n/2]} \sum_{b \in [1\dd \frac{1}{3} \log n]} I(i,b)$. We are to show that $\E[I(i,b)] > \frac{1}{8}$ when $n$ is large enough, which implies that $\E[I] \ge \Omega(n \log n)$ and, thus, the random string has $\Omega(n \log n)$ closed repeats in expectation.

Fix $i \in [2\dd n/2]$ and $b \in [1\dd \frac{1}{3} \log n]$. It is well known that $1 - 1/k \le e^{-1/k}$, for $k \ge 1$. Since the number of binary strings with length $b$ is $2^{b} \le \sqrt[3]{n}$, the probability that the string $s[i\dd i{+}b)$ has no occurrences at positions $i{+}b, i{+}2b, i{+}3b, \ldots$ (recall that $i \le n/2$) is at most $(1 - 1/2^b)^{\Theta(n/b)} \le (1 - 1/\sqrt[3]{n})^{\Theta(n/b)} \le (e^{-1/\sqrt[3]{n}})^{\Theta(n/b)} = e^{-\Theta(n^{2/3}/b)}$, which is upperbounded by $1 / e^{\Theta(\sqrt{n})}$ since $b < \log n$. Hence, with probability at least $1 - 1/e^{\Theta(\sqrt{n})}$, the string $s[i\dd i{+}b)$ has an occurrence at one of the positions $i{+}1, i{+}2,\ldots$ (the probability is taken for all random strings $s$ with the fixed length $n$ and the fixed $i$ and $b$). 

For integer $j > i$, denote by $E_j$ the following event: $s[i\dd i{+}b)$ occurs at position $j$ but does not occur at positions $i{+}1, i{+}2, \ldots, j{-}1$. Evidently, the events $E_j$ are mutually exclusive and $\sum_{j > i} \Pr(E_j)$ equals the probability that $s[i\dd i{+}b)$ occurs at one of the positions $i{+}1, i{+}2, \ldots$, which, as was shown above, is at least $1 - 1/e^{\Theta(\sqrt{n})}$. Denote by $C$ the event that $s[i\dd i{+}b)$ is a closed repeat. Conditioned on the event $E_j$, the event $C$ holds iff $s[i{-}1] \ne s[j{-}1]$ (recall that $i > 1$) and either $s[i{+}b] \ne s[j{+}b]$ or $j{+}b > n$. Since, provided $E_j$ happens, the letter $s[i{-}1]$ is chosen independently of other letters, we have $s[i{-}1] \ne s[j{-}1]$ with probability $\frac{1}2$; analogously, provided $E_j$ happens and $j{+}b \le n$, $s[j{+}b]$ is independent of other letters and, thus, $s[i{+}b] \ne s[j{+}b]$ with probability $\frac{1}{2}$. Therefore, $\Pr(C | E_j) = \frac{1}{2}\cdot\frac{1}{2} = \frac{1}{4}$, for $i < j \le n{-}b$, and $\Pr(C | E_j) = \frac{1}{2}$, for $j = n - b + 1$. Applying the formula of total probability, we lowerbound the probability that $s[i\dd i{+}b)$ is a closed repeat as follows:
\[
\Pr(C) = \sum_{j > i} \Pr(C | E_j)\cdot \Pr(E_j) \ge \frac{1}{4} \sum_{j > i} \Pr(E_j) \ge \frac{1}4 \left(1 - \frac{1}{e^{\Theta(\sqrt{n})}}\right).
\]
We obtain that $I(i,b) = 1$ with probability at least $\frac{1}{4} (1 - 1/e^{\Theta(\sqrt{n})})$, which is greater than $\frac{1}{8}$ when $n$ is sufficiently large. Thus, we obtain $\E[I(i,b)] > \frac{1}{8}$ as required.
\end{proof}

The definition of (right/left) closed repeats immediately implies the following lemma.

\begin{lemma}
	For a string $s$, let $s[i\dd i{+}t_1), \ldots, s[i\dd i{+}t_k)$ be all (right/left) closed repeats at position $i$ and let $s[q_1\dd q_1{+}t_1),\ldots,s[q_k\dd q_k{+}t_k)$ be their corresponding next occurrences. Suppose that $t_1 \le \cdots \le t_k$. Then, $q_1 < \cdots < q_k$ and $t_1 < \cdots < t_k$.\label{lem:occ-increase}
\end{lemma}

The following lemma somewhat details the character of growth for the sequence $q_1, \ldots,q_k$ in Lemma~\ref{lem:occ-increase}. 

\begin{lemma}
	For a string $s$, let $s[i\dd i{+}t_1),\ldots,s[i\dd i{+}t_k)$ be all (right) closed repeats at position $i$ and let  $s[q_1\dd q_1{+}t_1),\ldots,s[q_k\dd q_k{+}t_k)$ be their corresponding next occurrences. Suppose that $t_1 < \cdots < t_k$. Then, for any $x \in [1\dd k{-}2]$, we have $q_{x+2} > q_x + t_x$.\label{lem:next-occ}
\end{lemma}
\begin{proof}
	By Lemma~\ref{lem:occ-increase}, we have $q_1 < \cdots < q_k$. Assume that $q_{x+1} < q_x + t_x$ (otherwise there is nothing to prove). Since the string $s[q_x\dd q_x{+}t_x)$ occurs at position $q_{x+1}$, $p = q_{x+1} - q_x$ is a period of $s[q_x\dd q_x{+}t_x)$. Further, since $t_{x+1} > t_x$, the definition of right closed repeats implies that the letters $s[q_x{+}t_x]$ and $s[q_{x+1}{+}t_x] = s[q_x{+}p{+}t_x]$ are distinct. Thus, any substring with period $p$ cannot contain both  $s[q_x{+}t_x]$ and $s[q_x{+}p{+}t_x]$. In particular, $s[q_x\dd q_x{+}t_x)$ cannot occur at positions from $(q_x{+}p \dd q_x{+}t_x] = (q_{x+1}\dd q_x{+}t_x]$. Hence, $q_{x+2} > q_x + t_x$. 
\end{proof}

In some applications, it might be too expensive to spend $O(n\log n)$ space and, hence, one might desire to restrict the set of considered closed repeats to a somewhat representative subset of size $o(n\log n)$. For instance, it could be natural to consider only closed repeats with length larger than a given threshold $d$. The string $a^{n/2} b a^{n/2}$ shows that with any reasonable $d$ the number of closed repeats with length ${\ge}d$ still may be $\Omega(n)$. However, this example does not eliminate the hope to restrict the number of closed repeats to $O(n)$. 

Unfortunately, the right/left closed repeats cannot be meaningfully restricted in this way as the following examples over the alphabet $\{a,b,c\}$ show. For $d < \frac{1}{6}\log n$, the random binary string  provides $\Omega(n\log n)$ closed repeats with length ${\ge}d$ in expectation, which can be shown by slightly modifying the proof of Theorem~\ref{thm:lower}. For $d \ge \frac{1}{6}\log n$, consider the following string $y$:
\[
y = r_1 \overbrace{cc\cdots c}^{2d\text{ times}} r_2  \overbrace{cc\cdots c}^{2d\text{ times}} \ \ \cdot\ \cdot\ \cdot\ \  r_{k}\overbrace{cc\cdots c}^{2d\text{ times}} r_{k+1}
\]
where $k = \lfloor n / (\lceil\log n\rceil + 2d)\rfloor$, $r_1,r_2,\ldots,r_{k}$ are random strings over the alphabet $\{a,b\}$, each with length $\lceil\log n\rceil$, and $r_{k+1}$ is a random string over $\{a,b\}$ with length less than $\lceil\log n\rceil + 2d$ chosen so that the total length of $y$ is $n$. The expected number of right/left closed repeats with length ${\ge}d$ in this string is $\Omega(n \log n)$, provided $d \le n^{1 - \Omega(1)}$. For the proof, consider the set $Q$ of positions $i$ such that $y[i\dd i+d) = c^{d}$; we have $|Q| = \Omega(n)$. Fix an appropriate constant $\epsilon > 0$ and, for any $b \in [1\dd \epsilon\log n]$ and $i\in [1\dd n/2]\cap Q$, define the indicator random variable $I(i,b)$ that is equal to $1$ iff $y[i\dd i_b]$ is a right closed repeat (the case of left closed repeats is symmetric), where $i_b$ is the position such that $y[i\dd i_b] = c^{d'}r'$ with $d' \ge d$ and $r'$ is the prefix of length $b$ of the substring $r_{\lceil i / (\lceil\log n\rceil + 2d)\rceil}$. One can show that  $\E[I(i,b)] \ge \frac{1}{8}$, which implies that $\E[\sum_{i\in [1\dd n/2]\cap Q} \sum_{b\in [1\dd \epsilon \log n]} I(i,b)] \ge \Omega(n \log n)$. The details of this argument are essentially the same as in the proof of Theorem~\ref{thm:lower}, so, we omit them, especially because this same string $y$ is again similarly analyzed in more details below for a slightly more complicated case.

For the closed repeats, we consider only the thresholds $d \le \mathop{\mathrm{polylog}} n$. Unlike for the right/left closed repeats, we were unable to find tight lower and upper bounds for the number of closed repeats with length ${\ge}d$. The best known upper bound remains $O(n\log n)$. However, we managed to construct, for any $d = \Theta(\log^{c} n)$ with $c \ge 1$, strings in which the number of closed repeats with length ${\ge}d$ is $\Omega(n\log\log n + n\frac{\log^2 n}{d})$ (for the remaining thresholds $d < \frac{1}{6}\log n$, the random binary string provides $\Omega(n\log n)$ closed repeats with length ${\ge}d$ in expectation, as in Theorem~\ref{thm:lower}).

\begin{theorem}
	For any increasing function $d(n)$ such that $\Omega(\log n) \le d(n) \le \log^{O(1)} n$, there are arbitrarily long strings $x$ that contain $\Omega(n\log\log n + n\frac{\log^2 n}{d(n)})$ closed repeats with length ${\ge}d(n)$, where $n = |x|$.\label{thm:lower2}
\end{theorem}

\begin{proof}
	The strings will have the form $x = y\cdot z$, where $y$ is a string that contains $\Omega(n\frac{\log^2 n}{d(n)})$ closed repeats and $z$ contains $\Omega(n\log\log n)$ closed repeats.
	
	\textbf{\boldmath Strings $y$ with $\Omega(n\frac{\log^2 n}{d(n)})$ closed repeats.}
	For a given $n$, set $d = d(n)$ and define $y$ over the alphabet $\{a,b,c\}$ as follows:
	\[
	r_1 \overbrace{cc\cdots c}^{d\text{ times}} r_2  \overbrace{cc\cdots c}^{d\text{ times}} \ \ \cdot\ \cdot\ \cdot\ \  r_{k}\overbrace{cc\cdots c}^{d\text{ times}} r_{k+1}
	\]
	where $k = \lfloor n / (\lceil\log n\rceil + d)\rfloor$, $r_1,r_2,\ldots,r_{k}$ are random strings over the alphabet $\{a,b\}$ with length $\lceil\log n\rceil$, and $r_{k+1}$ is a random string over $\{a,b\}$ with length less than $\lceil\log n\rceil + d$ chosen so that the total length of $y$ is $n$. It suffices to prove that the expected number of closed repeats of length ${\ge}d$ in $y$ is $\Omega(n\frac{\log^2 n}{d})$, which implies that there exists a string $y$ with such number of closed repeats.
	
	Denote by $P$ the set of all positions in $y$ that are multiples of $\lceil\log n\rceil + d$, i.e., the positions of last letter $c$ in the runs $cc\cdots c$. Fix $i \in [1\dd n/2] \cap P$. For any $b_1 \in [1\dd \frac{1}{6}\log n]$ and $b_2 \in [1\dd \frac{1}{6}\log n]$, denote by $I(i,b_1,b_2)$ the indicator random variable that is equal to $1$ iff $y[i - d - b_1 + 1 \dd i + b_2]$ is a closed repeat. Then, the total number of closed repeats with length ${\ge}d$ in $y$ is greater than or equal to the following sum:
	\[
	I = \sum_{i \in [1\dd n/2]\cap P} \sum_{(b_1,b_2) \in [1\dd \frac{1}{6}\log n]^2}   I(i,b_1,b_2).
	\] 
	It suffices to show that $\E[I(i,b_1,b_2)] > \frac{1}{16}$ when $n$ is large enough, which implies $\E[I] \ge \Omega(n\frac{\log^2 n}{d})$.
	
	The analysis is very similar to the analysis in the proof of Theorem~\ref{thm:lower}. The number of binary strings with length $b_1 + b_2$ is at most $2^{b_1 + b_2} \le \sqrt[3]{n}$. Therefore, the probability that, for each $j \in (i\dd n] \cap P$, $y[i-d-b_1+1\dd i+b_2] \ne y[j-d-b_1+1\dd j+b_2]$ is at most $(1 - 1/2^{b_1 + b_2})^{\Theta(n / d)} \le (1 - 1/\sqrt[3]{n})^{\Theta(n / d)} \le (e^{-1/ \sqrt[3]{n}})^{\Theta(n / d)} = e^{-\Theta(n^{2/3} / d)}$, which is upperbounded by $e^{-\sqrt{n}}$ for sufficiently large $n$ since $d \le \log^{O(1)} n$. Hence, with probability at least $1 - e^{-\sqrt{n}}$ the string $y[i-d-b_1+1\dd i+b_2]$ has a next occurrence in $y$.
	
	The further analysis is the same as in the proof of Theorem~\ref{thm:lower}: conditioned on the event that the substring $y[i-d-b_1+1\dd i+b_2]$ has an occurrence at a position from $(i\dd n]$, we show that  $y[i-d-b_1+1\dd i+b_2]$ is a closed repeat with probability at least $\frac{1}{4}$ since the letters $y[i-d-b_1]$ and $y[i+b_2+1]$, which ``surround'' the potential closed repeat, are chosen independently from the corresponding ``surrounding'' letters of the next occurrence of the repeat. Then, we lowerbound the probability that $y[i-d-b_1+1\dd i+b_2]$ is a closed repeat by $\frac{1}{4}(1 - e^{-\sqrt{n}}) \ge \frac{1}{8}$. The argument can be easily made more careful and precise as in Theorem~\ref{thm:lower}, so, we omit the remaining details.

	
	\textbf{\boldmath Strings $z$ with $\Omega(n\log\log n)$ closed repeats.}
	Fix constant $c \ge 1$ such that $d(n) \le \log^c n$. We are to construct inductively a string $z$ with $\Omega(n\log\log n)$ closed repeats of length ${\ge}\log^c n$. Fix a sufficiently large integer $t > 5$ (the lower bound for $t$ depends on the constant $c$ and is implicitly determined in the sequel). For integer $q \ge 0$, suppose that we are given a string $s_q = a_1 a_2 \cdots a_{m}$ with $m \ge 2t$ such that:
	\begin{enumerate} 
		\item any substring of $s_q$ of length at least $\frac{1}{2}|s_q|$ occurs in $s_q$ only once;
		\item each prefix or suffix of $s_q$ with length $\ell \ge \frac{3}{4}|s_q|$ contains at least $q + \alpha q \ell$ closed repeats of length at least $t$, where $\alpha > 0$ is a constant that we will fix later. 
	\end{enumerate}
	Construct a string $s_{q+1} = s \cdot s'$, where $s$ and $s'$ are the following strings and $\$$ denotes unique distinct letters (all $\$$ are different):
	\begin{align*}
		&s = a_1 a_2 \cdots a_{m} \$ a_2 a_3 \cdots a_{m} \$ \ \ \cdot\ \cdot\ \cdot \ \ \$ a_{\lfloor m/4\rfloor} a_{\lfloor m/4\rfloor+1} \cdots a_{m} \$\\
		&s' = a_1 a_2 \cdots a_{\lceil 3m/4\rceil+1} \$ a_1 a_2 \cdots a_{\lceil 3m/4\rceil+2} \$ \ \ \cdot\ \cdot\ \cdot \ \ \$ a_1 a_2 \cdots a_{m} \$
	\end{align*}
	It is evident that any substring of $s_{q+1}$ of length at least $\frac{1}{2}|s_{q+1}|$ occurs in $s_{q+1}$ only once because of the sentinel letters $\$$. Consider a suffix (or prefix) $u$ of $s_{q+1}$ with length $\ell \ge \frac{3}{4}|s_{q+1}|$. We are to prove that $u$ contains at least $(q+1) + \alpha(q+1)\ell$ closed repeats with lengths at least $t$. We prove it when $u$ is a suffix; the case of prefix is symmetric and its analysis is the same. Thence, $s_{q+1}$ satisfies the above assumptions 1--2 and we show in the sequel how this implies the theorem. We count the closed repeats in $u$ in two groups: (a) and (b). 
	
	

	\textbf{\boldmath (a) Induction gives $(q+1) + \alpha q\ell$ closed repeats.}
	The suffix $u$ contains exactly $\lfloor m/4 \rfloor$ substrings of $s'$ of the form $a_1 a_2 \cdots a_j\$$ (since $u$ contains the whole $s'$), a number of whole substrings of $s$ of the form $\$a_ia_{i+1}\cdots a_{m}\$$, and, possibly, a proper suffix of one substring of the form $a_ia_{i+1}\cdots a_m\$$. The latter suffix must be a prefix of $u$ and we denote it by $u'$. Note that $|u'| \le m$. Let $u''$ be a string such that $u = u'u''$. By assumptions for the string $s_q$, any string $a_ia_{i+1}\cdots a_{m}$ or $a_1 a_2 \cdots a_{j}$ with length $k \ge \frac{3}{4} m$ contains at least $q + \alpha qk$ closed repeats of length ${\ge}t$, i.e., $\alpha q$ closed repeats per letter plus $q$ additional repeats. Assuming $\alpha \le \frac{1}{9}$, we can ``charge'' each letter of $u''$ with $\alpha q$ ``closed repeats per letter'', borrowing $\alpha q$ charges required for the letters $\$$ from the additional $q$ repeats in the corresponding substring $a_1 a_2 \cdots a_j\$$ or $a_ia_{i+1}\cdots a_{m}\$$. At least $\lfloor m/4\rfloor (1 - \alpha)q \ge \frac{1-\alpha}{8}qm$ charges remain after borrowing from the $q\lfloor m/4\rfloor$ additional repeats. Since $|u'| \le m$, this amount of charges suffices to charge the letters of $u'$ since $\frac{1 - \alpha}{8}qm \ge \alpha qm$ when $\alpha < \frac{1}{9}$. Still, at least $3q$ additional repeats remain from the part $s$ that were not charged to anything besides their corresponding letters $\$$, and they possess $3(1-\alpha)q \ge q + 1$ charges, for $q > 0$; for $q = 0$, we borrow one repeat from the group (b) described below. Thus, by the charging argument, we have overall obtained at least $(q + 1) + \alpha q\ell$ closed repeats of length ${\ge}t$ in $u$.
	
	\textbf{\boldmath (b) Non-inductive closed repeats give $\alpha \ell$.}
	Let us distinguish $\alpha\ell$ closed repeats in $u$ besides the ones obtained using the inductive assumptions on the string $s_q$. Since every substring of $s_q = a_1 a_2 \cdots a_{m}$ with length ${\ge}\frac{1}{2}m$ occurs in $s_q$ only once, it is straightforward to confirm that, for each substring $\$a_i a_{i+1}\cdots a_{m}\$$ of the first half $s$ of the string $s_{q+1}$ (with $i = 1,2,\ldots ,\lfloor m/4\rfloor$), each of its substring $a_i a_{i+1} \cdots a_j$ with $\lceil 3m/4\rceil < j \le m$ has length ${\ge}\frac{1}{2}m$ and is a closed repeat whose next occurrence is located in the second half $s'$ of the string $s_{q+1}$. 
	Therefore, each substring $\$a_ia_{i+1}\cdots a_m\$$ of the first half $s$ of $s_{q+1}$ contains at least $\lfloor m/4\rfloor$ closed repeats whose next occurrences are in $s'$.
	
	Observe that $|s| \ge \frac{3}{4}m \cdot \frac{1}{4}m \ge \frac{1}{8} m^2$. Then, the suffix $u$ contains at least $\frac{|s|}{2m} = \frac{1}{16} m$ substrings of $s$ of the form  $\$a_i a_{i+1}\cdots a_{m}\$$. Thus, $u$ contains at least $\lfloor m/4\rfloor \frac{1}{16} m \ge \frac{1}{128} m^2$ closed repeats of length ${\ge}\frac{1}{2}m$ described above. Recall that $m \ge 2t$, so that $\frac{1}{2}m \ge t$. Observe that $|s| = |s'|$ and $|s_{q+1}| = 2|s| = 2(\lfloor m/4\rfloor + \sum_{k=\lceil 3m/4\rceil+1}^{m} k) = \lfloor m/4\rfloor (2m - \lfloor m/4\rfloor + 3)$. Using straightforward calculations, we can upperbound this by $\frac{7}{16}m^2 + \frac{3}{4}m$. 
	We set $\alpha = \frac{1}{64}$ and deduce: $\alpha\ell \le \alpha|s_{q+1}| \le \alpha(\frac{7}{16}m^2 + \frac{3}{4}m) < \alpha \frac{8}{16}m^2 = \frac{1}{128}m^2$ (whenever $m \ge 12$, which holds since $m \ge 2t$ and $t > 5$). Thus, the $\frac{1}{128}m^2$ closed repeats obtained above suffice to charge $\alpha\ell$ repeats required for the suffix $u$.
	
	\textbf{The construction implies the theorem.} Since $|s_{q+1}| \le \frac{7}{16}m^2 + \frac{3}{4}m$, we obtain $|s_{q+1}| \le m^2 = |s_q|^2$ whenever $m \ge 2$. Our recursive construction starts with the string $s_0$ of length $m = 2t$  consisting of all distinct letters. The assumptions 1--2 obviously hold for $s_0$. Let $k = \lceil\frac{1}{2c}\log t\rceil$. We consecutively obtain the strings $s_1, s_2, \ldots, s_k$ and our resulting string is $s_k$. Denote $n = |s_k|$. The string $s_k$ contains at least $\alpha k n$ closed repeats with lengths at least $t$. We have $n \le |s_0|^{2^k} = (2t)^{2^k}$ due to the inequality $|s_{q+1}| \le |s_q|^2$. Hence, $\log n \le 2^k \log (2t) \le 2^{k}\cdot 2 \log t \le 2^{k}\cdot 4ck$. Solving the inequality $\log n \le 4c\cdot 2^k k$, we deduce $k = \Omega(\log\log n)$ (the dependency on the constant $c$ is hidden under the $\Omega$). Thus, $s_k$ contains $\alpha n k = \Omega(n\log\log n)$ closed repeats with lengths ${\ge}t$. Further, we deduce $\log n \le 2^{k} \cdot 2\log t \le 2^{1 + \frac{1}{2c}\log t}2\log t = 4\cdot t^{\frac{1}{2c}} \log t$. For $t \le \log^c n$, it implies that $\log n \le 4 c \sqrt{\log n}\log\log n$, which is impossible whenever $n$ is large enough. Observe that larger values for $t$ produce larger $n$. Hence, we have $t > \log^c n$ whenever $t$ is sufficiently large (this is the place where $t$ depends on $c$). Thus, $s_k$ can be used as the string $z$ in the decomposition $x = y\cdot z$ as required and we can make $s_k$ arbitrarily long by starting with arbitrarily large values for $t$.
\end{proof}

One may propose another possible approach to restrict the number of closed repeats to a meaningful subset: consider only closed repeats that start at a given set of positions $R \subseteq [1\dd n]$ of the string $s[1\dd n]$. After a closer inspection of the string $z$ constructed in the proof of Theorem~\ref{thm:lower2}, one can see that all considered closed repeats in $z$ started at $O(n / d(n))$ positions. Therefore, when $|R| = \Theta(\frac{n}{d})$ with $d \le \log^{O(1)} n$, one cannot hope to have less than $\Omega(n\log\log n)$ closed repeats starting at positions from $R$.

\section{Computation}
\label{sec:computation}

In~\cite{BDFP2} an algorithm was proposed that computes all maximal closed substrings in  a string of length $n$ in $O(n\log n)$ time using the suffix tree. On closer inspection, it is straightforward that this same algorithm also computes all closed repeats, due to the correspondence outlined in Section~\ref{sec:closed-repeats}, and all right/left closed repeats (in fact, it computes precisely the right closed repeats but filters out those of them that are not left closed). For the sake of brevity, we do not describe this simple modification here and refer the reader to~\cite{BDFP2}.

While the time $O(n\log n)$ is optimal in the worst case when there are $\Theta(n\log n)$ closed repeats in the string, it is natural to ask whether a better output-dependent algorithm with $O(n + c)$ time is possible, where $c$ is the number of (right/left) closed repeats. We investigate this problem under the assumption of \emph{general ordered alphabet} where the alphabet is totally ordered and the only operation on two letters is the comparison that evaluates their relative order in $O(1)$ time.

As it will be seen in Section~\ref{sec:applications}, the right closed repeats can be used to compute the so-called LZ77 parsing, which requires $\Omega(n\log\sigma)$ time on the general ordered alphabet of size $\sigma$, as was proved in~\cite{EllertPhD,KosolobovLZRuns}. This observation suggests that any algorithm computing the closed repeats must contain a kind of indexing data structure, like the suffix tree, that takes $\Omega(n\log\sigma)$ time to construct on the general alphabet. Unfortunately, the ``hard'' strings for the LZ77 parsing that were used in~\cite{EllertPhD,KosolobovLZRuns} contain $\Omega(n\log n)$ closed repeats in the worst case. Hence, we prove the lower bound $\Omega(n\log\sigma)$ using a different reduction.


\begin{theorem}
	Any algorithm computing all closed repeats in strings of length $n$ over a general ordered alphabet of size $\Theta(\sigma)$, for $\sigma \ge n^{\Omega(1)}$, spends $\Omega(n\log n)$ time on at least one string containing $O(n)$ closed repeats.\label{thm:lower-bound-alphabet}%
\end{theorem}
\begin{proof}
	We devise a reduction from the \emph{Alphabet Set Testing} problem defined by Ellert~\cite{EllertPhD}: given two disjoint alphabets $\Sigma = \{a_1,\ldots,a_\sigma\}$, $\Sigma' = \{b_1,\ldots,b_\sigma\}$, and a string $x = x_1 x_2\cdots x_{\bar{n}} \in (\Sigma\cup\Sigma')^{\bar{n}}$, one must decide whether $x \in \Sigma^{\bar{n}}$. It was proved in~\cite{EllertPhD} that this problem requires $\Omega(\bar{n}\log\sigma)$ time in the worst case on the general alphabet. Let $\sigma \ge n^{1/(d-1)}$, for constant $d \ge 2$. Consider $s = b_1 \cdots b_\sigma x_1 Z_1 x_2 Z_2 \cdots x_{\bar{n}} Z_{\bar{n}} a_1\cdots a_\sigma$, where $Z_1, Z_2, \ldots, Z_{\bar{n}}$ are distinct strings, each of length $d-1$, over a new alphabet $\{\$_1,\ldots,\$_\sigma\}$. We will put $\bar{n} \le n$ so that $\sigma^{d-1} \ge n \ge \bar{n}$ to guarantee that there are at least $\bar{n}$ distinct strings $Z_i$ of length $d-1$. Since $x \in \Sigma^{\bar{n}}$ iff each substring $x_i$ in $s$ is a closed repeat, the algorithm computing all closed repeats for $s$ solves the Alphabet Set Testing. Since all $Z_1,\ldots,Z_{\bar{n}}$ are distinct, $s$ contains less than $2d|s| = O(|s|)$ closed repeats. Further, we have $|s| = d\bar{n} + 2\sigma$ so that $|s| = n$ whenever $\bar{n} = n/d - 2\sigma/d$. We assume that $\sigma \le n/4$ and set $\bar{n} = \lfloor n/d - 2\sigma/d\rfloor$, padding the string $s$ arbitrarily to make its length $n$. Since $\sigma \le n/4$, we have $n \ge \bar{n} \ge \lfloor n /(2d)\rfloor = \Omega(n)$. Thus, we obtain for the algorithm the lower bound $\Omega(\bar{n}\log\sigma) = \Omega(n\log\sigma) = \Omega(n\log n)$ attained on a string over an alphabet of size $\Theta(\sigma)$ with $O(n)$ closed repeats.
\end{proof}

Theorem~\ref{thm:lower-bound-alphabet} contrasts with the known results about runs, which are a particular case of closed repeats: all runs can be computed in $O(n)$ time on the general ordered alphabet, as was shown by Ellert and Fischer~\cite{EllertFischer}. It remains open whether the (right/left) closed repeats can be computed in $O(n + c)$ time over a linearly-sortable alphabet, where $c$ is the number of closed repeats, or in $O(n\log\sigma + c)$ time over the general ordered alphabet of size $\sigma$.

\section{Applications}
\label{sec:applications}

In this section we offer a number of applications for closed repeats. They all describe data structures with $O(n\log n)$ space that support a certain type of queries on substrings. Not all of them are strictly novel but we believe that our solutions are conceptually simpler than the state of the art, even if slightly inferior sometimes. We start with the following problem.


\begin{mdframed}
\begin{2-period-queries}
	Given a string $s$ of length $n$, construct a data structure with $O(n\log n)$ space that, for any substring $s[i\dd i{+}\ell)$, decides in $O(1)$ time whether it has a period $p$ such that $2p \le \ell$ and calculates $p$ if so.
\end{2-period-queries}
\end{mdframed}

This problem was introduced in \cite{CIKRRW}. In~\cite{KRRW} a solution with $O(1)$ query time and $O(n\log n)$ space was presented, based on an approach different from ours. In~\cite{KRRW2} a complicated data structure was used to reduce the space to $O(n)$, which beats our solution.

Observe that $p$ is a period of $s[i\dd i{+}\ell)$ iff $s[i{+}p\dd i{+}\ell)$ is a border of $s[i\dd i{+}\ell)$ (i.e., $s[i{+}p\dd i{+}\ell)$ is a prefix and suffix of $s[i\dd i{+}\ell)$). Thus, the problem is to find the longest border provided its length is at least $\ell/2$. We relate the borders to closed repeats as follows. If $s[i\dd i{+}t)$ is a right closed repeat whose next occurrence $s[q\dd q{+}t)$ is such that $q < i + \ell \le q + t$, then $s[q\dd i{+}\ell)$ is a border of $s[i\dd i{+}\ell)$. As a partial converse, if $s[i{+}p \dd i{+}\ell)$ is the longest border of $s[i\dd i{+}\ell)$ with length at least $\ell / 2$ (thus, $p$ is the minimal period), then there is a right closed repeat $s[i\dd i{+}t)$ whose next occurrence is at position $i + p$ and $i + \ell \le i + p + t$; this claim follows from the observation that $s[i{+}p \dd i{+}\ell)$ cannot occur at positions between $i$ and $i+p$ since otherwise the substring $s[i\dd i{+}p)$ would be equal to a non-trivial cyclic rotation of itself (see Fig.~\ref{fig:border-self-cover}), which would imply that $s[i\dd i{+}p) = v^k$ for some $v$ and $k > 1$ (see~\cite{Lothaire}), thus contradicting the minimality of $p$.

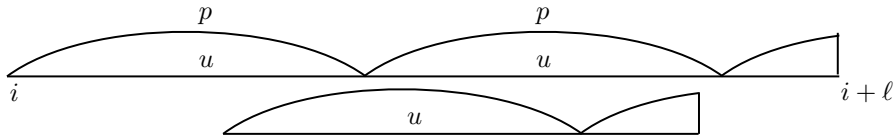
\begin{figure}[hbt]
	\tikzset{every picture/.style={line width=0.75pt}} 
	
	\begin{tikzpicture}[x=0.75pt,y=0.75pt,yscale=-1,xscale=1]
		
		\draw    (29.7,50.13) -- (449.7,50.13) ;
		\draw    (29.7,50.13) .. controls (69.7,20.13) and (169.7,20.13) .. (210.7,50.13) ;
		\draw    (209.7,50.13) .. controls (249.7,20.13) and (349.7,20.13) .. (390.7,50.13) ;
		\draw    (389.7,50.13) .. controls (429.7,20.13) and (529.7,20.13) .. (570.7,50.13) ;
		\draw  [color={rgb, 255:red, 255; green, 255; blue, 255 }  ,draw opacity=1 ][fill={rgb, 255:red, 255; green, 255; blue, 255 }  ,fill opacity=1 ] (450,26.13) -- (575.7,26.13) -- (575.7,79.13) -- (450,79.13) -- cycle ;
		\draw    (138.7,79.13) -- (378.7,79.13) ;
		\draw    (138.7,79.13) .. controls (178.7,49.13) and (278.7,49.13) .. (319.7,79.13) ;
		\draw    (318.7,79.13) .. controls (358.7,49.13) and (458.7,49.13) .. (499.7,79.13) ;
		\draw  [color={rgb, 255:red, 255; green, 255; blue, 255 }  ,draw opacity=1 ][fill={rgb, 255:red, 255; green, 255; blue, 255 }  ,fill opacity=1 ] (379,52) -- (516.7,52) -- (516.7,88.13) -- (379,88.13) -- cycle ;
		\draw    (378.7,58.13) -- (378.7,79.13) ;
		\draw    (448.7,28.13) -- (448.7,49.13) ;
		
		\draw (125,14) node [anchor=north west][inner sep=0.75pt]    {$p$};
		\draw (295,14) node [anchor=north west][inner sep=0.75pt]    {$p$};
		\draw (125,38) node [anchor=north west][inner sep=0.75pt]    {$u$};
		\draw (295,38) node [anchor=north west][inner sep=0.75pt]    {$u$};
		\draw (230,67) node [anchor=north west][inner sep=0.75pt]    {$u$};
		\draw (29.7,52.13) node [anchor=north west][inner sep=0.75pt]    {$i$};
		\draw (448.7,51.13) node [anchor=north west][inner sep=0.75pt]    {$i+\ell $};

	\end{tikzpicture}
	
	\caption{A schematic depiction of a string $s[i\dd i{+}\ell)$ with a period $p$ such that $2p \le \ell$. The ``hills'' depict the substrings $u = s[i\dd i{+}p) = s[i{+}p\dd i{+}2p) =\cdots$ As is shown, whenever the border $s[i{+}p\dd i{+}\ell)$ occurs at a position between $i$ and $i + p$, the string $u$ must be equal to a non-trivial cyclic rotation of itself.}\label{fig:border-self-cover}
\end{figure}

In our data structure, for each $i \in [1\dd n]$, we associate with $i$ the list of all right closed repeats starting at position $i$, stored in the increasing order of lengths; for each of the repeats, we also store its next occurrence. Further, for each $i$ and each $d \in [1\dd\log n]$, we store the pointer to the longest right closed repeat in the list with length at most $2^d$ (if any). Let us show how to perform the query on a substring $s[i\dd i{+}\ell)$. 

Denote by $t_1,\ldots, t_k$ and $q_1,\ldots, q_k$, respectively, the lengths and the positions of the next occurrences of all closed repeats in the list associated with $i$. By Lemma~\ref{lem:occ-increase}, we have $t_1 < \cdots < t_k$ and $q_1 < \cdots < q_k$.  Due to the above discussion, if the minimal period $p$ of $s[i\dd i{+}\ell)$ is such that $2p \le \ell$, then $p = q_x - i$ where $q_x$ is the smallest position in the list such that $i + \ell \le q_x + t_x$. We first compute in $O(1)$ time $d = \max\{d\in\mathbb{Z} \colon 2^d < \ell\}$ using either a specialized instruction or a precomputed table of size $n$. Then, using stored pointers, we compute $x' = \arg\max_{x\in[1\dd k]}\{t_x \le 2^{d-1}\}$ or $x' = 0$ if $t_x > 2^{d-1}$ for all $x$. We consecutively check the condition $i + \ell \le q_x + t_x$ for $x = x' + 1, x' + 2, \ldots$ We claim that it suffices to consider only $x = x' + 1$, $x' + 2$, $x' + 3$, $x' + 4$ and there is no period $p$ such that $2p \le \ell$ if the condition is not satisfied for these $x$.

Clearly, it makes sense to consider only such $x \in [1\dd k]$ for which $q_x \le i + \ell / 2$ and $t_x \ge \ell / 2$. The inequalities $t_{x'} \le 2^{d-1}$ and $2^{d-1} < \ell/2$ imply that $x$ must be greater than $x'$ to guarantee $t_x \ge \ell / 2$. Since $2^d \ge \ell / 2$, it remains to prove that $q_{x'+5} > i + 2^d$, which implies $q_{x'+5} > i + \ell/2$. By Lemma~\ref{lem:next-occ}, we obtain $q_{x'+3} > q_{x'+1} + t_{x'+1} > i + 2^{d-1}$ and, again due to Lemma~\ref{lem:next-occ}, $q_{x'+5} > q_{x'+3} + t_{x'+3} > i + 2^{d-1} + 2^{d-1} = i + 2^d$.



\begin{mdframed}
\begin{longest-substring-repeat}
	Given a string $s$ of length $n$ and a constant $\epsilon > 0$, construct a data structure with $O(n\log n)$ space that, for any substring $s[i\dd i{+}\ell)$, can compute its longest repeat in $O(\log^\epsilon n)$ time (any of the longest repeats if there are more than one possible answers).
\end{longest-substring-repeat}
\end{mdframed}

Despite its natural statement, it seems that this problem was not previously investigated.

Given a substring $s[i\dd i{+}\ell)$, suppose that its longest repeat has length $t$, and $s[p\dd p{+}t)$ and $s[q\dd q{+}t)$ are two leftmost occurrences of this repeat in  $s[i\dd i{+}\ell)$ so that $i \le p < q < q + t \le i + \ell$. Evidently, if $q + t < i + \ell$, then $s[p\dd p{+}t)$ is a right closed repeat in $s$ and $s[q\dd q{+}t)$ is its next occurrence. We build the following data structure to detect such repeats: for each closed repeat $s[p\dd p{+}t)$ in $s$ and its next occurrence $s[q\dd q{+}t)$, we form a 2D point $(p, q + t)$ with weight $t$; we then construct on all these points the orthogonal range maxima data structure by Nekrich~\cite{Nekrich} that, given a constant $\epsilon > 0$, takes linear space (in our case $O(n\log n)$ since we have $O(n\log n)$ points) and can report, for any range $[a,\infty]\times [-\infty, b]$, a point from the range with maximal weight in $O(\log^{\epsilon} n)$ time. Thus, the query on the range $[i,\infty]\times [-\infty, i{+}\ell]$ will find the longest repeat in the substring $s[i\dd i{+}\ell)$ provided this repeat occurs in $s[i\dd i{+}\ell{-}1)$ twice.

Symmetrically, if $s[p\dd p{+}t)$ and $s[q\dd q{+}t)$ are two rightmost occurrences of a longest repeat in $s[i\dd i{+}\ell)$ and $i < p < q < q + t \le i + \ell$ (note that $i < p$), then the substring of the reversed string $\cev{s} = s[n]\cdots s[2]s[1]$ that corresponds to  $s[q\dd q{+}t)$ is a right closed repeat in $\cev{s}$ whose next occurrence is the substring corresponding to  $s[p\dd p{+}t)$ and it can be detected by an analogous data structure on the reversed string $\cev{s}$.

It remains to consider the case when the longest repeat in $s[i\dd i{+}\ell)$ is a prefix of $s[i\dd i{+}\ell)$ and its next occurrence is a suffix of $s[i\dd i{+}\ell)$, i.e., it is the longest border that does not occur anywhere else in  $s[i\dd i{+}\ell)$. Let $s[i\dd i{+}t_1), \ldots, s[i\dd i{+}t_k)$ be the list of all right closed repeats at position $i$ such that $t_1 < \cdots < t_k$ and let $s[q_1\dd q_1{+}t_1), \ldots, s[q_k\dd q_k{+}t_k)$ be their respective next occurrences. By Lemma~\ref{lem:occ-increase}, $q_1 < \cdots < q_k$. Evidently, the described longest border is equal to $s[q_x\dd i{+}\ell)$ where $x = \arg\min_{x\in [1\dd k]}\{q_x + t_x \ge i + \ell\}$. By equipping the list $q_1 + t_1,\ldots, q_k + t_k$ with the van Emde Boas structure~\cite{vanEmdeBoas}, one can compute $x$ in $O(\log\log n)$ time.

Thus, the data structure obtains three candidates for the longest repeat in $s[i\dd i{+}\ell)$, according to the three described cases, and it reports the longest one of them.

\medspace

The next problem requires the following definitions for its statement. For a given string $s$, its \emph{LZ77 parsing} (Lempel--Ziv \cite{LZ77}) is the decomposition $s = f_1 f_2\cdots f_z$ built from left to right by the following greedy procedure: if a prefix $s[1\dd p) = f_1 f_2 \cdots f_{m-1}$ is already processed, then the string $f_m$ (called a \emph{phrase}) is either the letter $s[p]$ that does not occur in $s[1\dd p)$ or is the longest substring that starts at position $p$ and has an occurrence at position $q < p$. We call the parsing \emph{rightmost} if, together with each phrase $f_m$, it provides the maximal position $q < p$ at which $f_m$ occurs (if any).

\begin{mdframed}
\begin{substring-compression-queries}
	Given a string $s$ of length $n$, construct a data structure with $O(n\log n)$ space that can compute the rightmost LZ77 parsing of any substring $s[i\dd i{+}\ell)$ in $O(z\log\log\ell)$ time, where $z$ is the number of phrases in the parsing.
\end{substring-compression-queries}
\end{mdframed}

The LZ77 parsing~\cite{LZ77} is one of the key structures in data compression. Its rightmost variant is in focus of research \cite{AmirLandauUkkonen,BelazzouguiPuglisi,EllertFischerPedersen,Larsson,SumiyoshiMienoInenaga} due to its nice properties~\cite{FerraginaNittoVenturini,KosolobovBitOpt,KVNP}. Substring compression queries were first considered in \cite{CormodeMuthukrishnan2}. In~\cite{KKFL} a data structure was proposed with $O(z\log\log n)$ query time and $O(n\log\log n)$ space (the time/space is like this when an up-to-date predecessor data structure~\cite{Zhou} is used in it; see also \cite{KRRW2}). However, this data structure does not necessarily return the rightmost parsing. Apparently, our solution is the first one for the rightmost LZ77. Also, our query time $O(z\log\log\ell)$ is faster than $O(z\log\log n)$ when $\ell \le 2^{2^{o(\log\log n)}}$. We pay for this by spending more space, $O(n\log n)$.

For our data structure, we associate with each $p \in [1\dd n]$ the list of all right closed repeats whose next occurrence is at position $p$; we store these repeats in the increasing order of lengths. Let $s[q_1\dd q_1{+}t_1), \ldots, s[q_k\dd q_k{+}t_k)$ be all repeats in the list associated with $p$ and $t_1 \le \cdots \le t_k$. Since each right closed repeat $s[q_x\dd q_x{+}t_x)$ cannot occur at positions between $q_x$ and $p$, we have $q_1 > \cdots > q_k$ and $t_1 < \cdots < t_k$ (cf. Lemma~\ref{lem:occ-increase}). This observation suggests that all these lists can be straightforwardly constructed in one left-to-right pass on all right closed repeats. We equip the list $q_1,\ldots,q_k$ with a predecessor data structure that, for any $i$, can compute $\min\{q_x \ge i\}$ in $O(\log\log(p - i))$ time: to this end, the list is split into $\log n$ sublists such that, for $t = 1,2,\ldots$, the $t$th sublist contains only positions $q_x$ with $2^{t-1} \le p - q_x < 2^t$ and it is equipped with the van Emde Boas structure~\cite{vanEmdeBoas}; we also store $\max\{p - q_x \le 2^{t-1}\}$ (if any) with the $t$th sublist. Thus, for any $i < p$, the query for $\min\{q_x \ge i\}$ first determines in $O(1)$ time $t$ such that $2^{t-1} \le p - i < 2^t$ (using either a specialized instruction or a precomputed table of size $n$) and then computes the answer in $O(\log\log 2^t) =  O(\log\log(p - i))$ time using the van Emde Boas structure on the $t$th sublist and the stored number $\max\{p - q_x \le 2^{t-1}\}$ in case the answer is not in the sublist. We equip the list $t_1,\ldots,t_k$ with an analogous data structure that, for any $\ell$, can compute $\min\{t_x \ge \ell\}$ in $O(\log\log\ell)$ time.

Given a substring $s[i\dd i{+}\ell)$, we construct its LZ77 parsing from left to right. Suppose that a prefix $s[i\dd i{+}p)$ is already parsed into phrases $f_1 f_2 \cdots f_{m-1}$ and we are to find the phrase $f_m$ that starts at position $p$. Assume that the letter $s[p]$ occurs in $s[i\dd p)$. Let $s[q_1\dd q_1{+}t_1), \ldots, s[q_k\dd q_k{+}t_k)$ be all repeats in the list associated with $p$. Denote by $s[q\dd q{+}|f_m|)$ the rightmost occurrence of $f_m$ preceding position $p$. Clearly, $s[q\dd q{+}|f_m|)$ is a prefix of a right closed repeat and its next occurrence is at position $p$; moreover, if $f_m$ is not the last phrase in the parsing, then $s[q\dd q{+}|f_m|)$ is itself the right closed repeat. Hence, $q$ must be in the list $q_1,\ldots,q_k$. Further, $i \le q < p$ and, provided $f_m$ is not the last phrase, the list cannot contain a closed repeat $s[q_x \dd q_x{+}t_x)$ such that $i \le q_x < q$: if it were the case, then $t_x > |f_m|$ and there would be a substring of length $t_x$ at position $p$ that has an occurrence at $q_x$. Thus, using a predecessor query, we compute in $O(\log\log\ell)$ time $q = q_x$ and $|f_m| = t_x$ where $x = \arg\min_{x\in [1\dd k]}\{q_x \ge i\}$. Observe that $q_1 < i$ iff $s[p]$ does not occur in $s[i\dd p)$; we put $f_m = s[p]$ in this case. One caveat remains: if thus calculated length $t_x$ exceeds $i + \ell - p$, then we ``truncate'' $f_m$ to length $i + \ell - p$ so that $f_m$ is last in the parsing. However, in this case, the found ``source'' $q = q_x$ for $f_m$ is not necessarily the rightmost one. We compute the rightmost ``source'' $q = q_y$ for the last phrase by one predecessor query $y = \arg\min_{y\in [1\dd k]}\{t_y \ge |f_m|\}$.

\section{Open Problems}
\label{sec:conclusion}

The presented results naturally inspire a number of open problems.

\begin{enumerate}
	\item Can we find all (right/left) closed repeats in $O(n\log\sigma + c)$ time over a general ordered alphabet of size $\sigma$, where $c$ is the number of found closed repeats?
	\item Can we reduce the time to $O(n + c)$ on linearly-sortable alphabets? A partial progress on this question has been recently achieved in \cite{FischerKoppl}.
	\item Can the query time in the longest substring repeat problem be $O(\log\log n)$? It is natural to expect the improvement due to the specificity of range queries in the current solution.
	\item What is the tight lower and upper bound for the number of closed repeats with length ${\ge}d$?
	\item What is the tight lower and upper bound for the number of closed repeats starting at a given set of positions $R$ of size $\Theta(\frac{n}{d})$?
	\item Can the space be reduced to $O(n)$ in the problem of longest substring repeat and the problem of substring compression queries (with the rightmost LZ77 parsing)?
	\item The algorithm from \cite{BDFP2} that can be adapted to compute the (right/left) closed repeats is practical and relatively easily implementable. However, it requires the construction of the suffix tree for its work, which is not a problem \emph{per se} but it suggests that, maybe, there can be a more direct algorithm that constructs the closed repeats without first building an indexing data structure like the suffix tree. We note, however, that, due to Theorem~\ref{thm:lower-bound-alphabet}, a certain kind of indexing is still unavoidable for $O(n\log\sigma)$ time, as in algorithms from \cite{BDFP2,ShibataEtAl,JainMhaskar}. A very different approach has been recently introduced in an $O(n\log n)$-time algorithm from \cite{JainMhaskar} based on \cite{Crochemore2}.
\end{enumerate}

\paragraph{Acknowledgement}
The author wishes to thank Gabriele Fici for the proposal of the problem about the number of maximal closed substrings at StringMasters 2024 in Fukuoka, Japan, which led to the present work, and for verifying the proof of Theorem~\ref{thm:lower}, which solved the proposed problem.

\paragraph{Declarations}

This work was carried out with the financial support of the Ministry of Science and Higher Education of the Russian Federation (theme no.~FEUZ-2026-0007).

\bibliographystyle{elsarticle-num-names-alphsort}
\bibliography{refs}

\begin{thebibliography}{10}
\expandafter\ifx\csname url\endcsname\relax
  \def\url#1{\texttt{#1}}\fi
\expandafter\ifx\csname urlprefix\endcsname\relax\def\urlprefix{URL }\fi

\bibitem{AmirLandauUkkonen}
A.~Amir, G.~M. Landau, E.~Ukkonen, Online timestamped text indexing,
  Information Processing Letters 82~(5) (2002) 253--259.

\bibitem{BadkobehCrochemoreToopsuwan}
G.~Badkobeh, M.~Crochemore, C.~Toopsuwan, Computing the maximal-exponent
  repeats of an overlap-free string in linear time, in: Proc. String Processing
  and Information Retrieval (SPIRE), vol. 7608 of LNCS, Springer, 2012.

\bibitem{BDFP}
G.~Badkobeh, A.~De~Luca, G.~Fici, S.~Puglisi, Maximal closed substrings, in:
  Proc. String Processing and Information Retrieval (SPIRE), vol. 13617 of
  LNCS, Springer, 2022.

\bibitem{BDFP2}
G.~Badkobeh, A.~De~Luca, G.~Fici, S.~Puglisi, Finding maximal closed
  substrings, Theoretical Computer Science 1060~(C) (2025) 115628.

\bibitem{BadkobehFiciLiptak}
G.~Badkobeh, G.~Fici, Z.~Lipt{\'a}k, On the number of closed factors in a word,
  in: Proc. Language and Automata Theory and Applications (LATA), vol. 8977 of
  LNCS, Springer, 2015.

\bibitem{BannaiEtAl}
H.~Bannai, T.~I, S.~Inenaga, Y.~Nakashima, M.~Takeda, K.~Tsuruta, The
  “runs” theorem, SIAM Journal on Computing 46~(5) (2017) 1501--1514.

\bibitem{BelazzouguiPuglisi}
D.~Belazzougui, S.~J. Puglisi, Range predecessor and {L}empel--{Z}iv parsing,
  in: Proc. the ACM-SIAM Symposium on Discrete Algorithms (SODA), SIAM, 2016.

\bibitem{BlumerEtAl}
A.~Blumer, J.~Blumer, D.~Haussler, R.~McConnell, A.~Ehrenfeucht, Complete
  inverted files for efficient text retrieval and analysis, Journal of the ACM
  34~(3) (1987) 578--595.

\bibitem{BLPS}
G.~S. Brodal, R.~B. Lyngs{\o}, C.~N. Pedersen, J.~Stoye, Finding maximal pairs
  with bounded gap, in: Proc. Combinatorial Pattern Matching (CPM), vol. 1645
  of LNCS, Springer, 1999.

\bibitem{CormodeMuthukrishnan2}
G.~Cormode, S.~Muthukrishnan, Substring compression problems., in: Proc. the
  ACM-SIAM Symposium on Discrete Algorithms (SODA), vol.~5, Citeseer, 2005.

\bibitem{Crochemore2}
M.~Crochemore, An optimal algorithm for computing the repetitions in a word,
  Information Processing Letters 12~(5) (1981) 244--250.

\bibitem{Crochemore}
M.~Crochemore, Transducers and repetitions, Theoretical Computer Science 45
  (1986) 63--86.

\bibitem{CIKRRW}
M.~Crochemore, C.~S. Iliopoulos, M.~Kubica, J.~Radoszewski, W.~Rytter,
  T.~Wale{\'n}, Extracting powers and periods in a word from its runs
  structure, Theoretical Computer Science 521 (2014) 29--41.

\bibitem{CrochemoreKolpakovKucherov}
M.~Crochemore, R.~Kolpakov, G.~Kucherov, Optimal bounds for computing
  $\alpha$-gapped repeats, Information and Computation 268 (2019) 104434.

\bibitem{EllertPhD}
J.~Ellert, Efficient string algorithmics across alphabet realms, {PhD} thesis,
  der Technischen Universität Dortmund an der Fakultät für Informatik
  (2024).

\bibitem{EllertFischer}
J.~Ellert, J.~Fischer, Linear time runs over general ordered alphabets, in:
  N.~Bansal, E.~Merelli, J.~Worrell (eds.), Proc. International Colloquium on
  Automata, Languages, and Programming (ICALP), vol. 198 of Leibniz
  International Proceedings in Informatics (LIPIcs), Schloss Dagstuhl --
  Leibniz-Zentrum f{\"u}r Informatik, Dagstuhl, Germany, 2021.

\bibitem{EllertFischerPedersen}
J.~Ellert, J.~Fischer, M.~R. Pedersen, New advances in rightmost
  {L}empel-{Z}iv, in: Proc. String Processing and Information Retrieval
  (SPIRE), vol. 14240 of LNCS, Springer, 2023.

\bibitem{FerraginaNittoVenturini}
P.~Ferragina, I.~Nitto, R.~Venturini, On the bit-complexity of {L}empel--{Z}iv
  compression, SIAM Journal on Computing 42~(4) (2013) 1521--1541.

\bibitem{Fici}
G.~Fici, et~al., Open and closed words, Bulletin of EATCS 3~(123) (2017) 1--11.

\bibitem{FischerKoppl}
J.~Fischer, D.~K{\"o}ppl, Output-sensitive enumeration of maximal closed
  substrings in linear space, arXiv preprint arXiv:2610.38114 (2026) 1--11.

\bibitem{GIIKM}
P.~Gawrychowski, T.~I, S.~Inenaga, D.~K\"{o}ppl, F.~Manea, Efficiently finding
  all maximal $\alpha$-gapped repeats, in: Proc. Symposium on Theoretical
  Aspects of Computer Science (STACS), vol.~47 of Leibniz International
  Proceedings in Informatics (LIPIcs), Schloss Dagstuhl -- Leibniz-Zentrum
  f{\"u}r Informatik, 2016.

\bibitem{GIIKM2}
P.~Gawrychowski, T.~I, S.~Inenaga, D.~K{\"o}ppl, F.~Manea, Tighter bounds and
  optimal algorithms for all maximal $\alpha$-gapped repeats and palindromes:
  finding all maximal $\alpha$-gapped repeats and palindromes in optimal worst
  case time on integer alphabets, Theory of Computing Systems 62 (2018)
  162--191.

\bibitem{Gusfield}
D.~Gusfield, Algorithms on strings, trees and sequences: computer science and
  computational biology, Cambridge university press, 1997.

\bibitem{IKoppl}
T.~I, D.~K{\"o}ppl, Improved upper bounds on all maximal $\alpha$-gapped
  repeats and palindromes, Theoretical Computer Science 753 (2019) 1--15.

\bibitem{JainMhaskar}
S.~K. Jain, N.~Mhaskar, Efficient algorithms to compute closed substrings,
  Theory of Computing Systems 70~(3) (2026) 48.

\bibitem{KKFL}
O.~Keller, T.~Kopelowitz, S.~L. Feibish, M.~Lewenstein, Generalized substring
  compression, Theoretical Computer Science 525 (2014) 42--54.

\bibitem{KRRW2}
T.~Kociumaka, J.~Radoszewski, W.~Rytter, T.~Wale{\'n}, Efficient data
  structures for the factor periodicity problem, in: Proc. String Processing
  and Information Retrieval (SPIRE), vol. 7608 of LNCS, Springer, 2012.

\bibitem{KRRW}
T.~Kociumaka, J.~Radoszewski, W.~Rytter, T.~Wale{\'n}, Internal pattern
  matching queries in a text and applications, in: Proc. the ACM-SIAM Symposium
  on Discrete algorithms (SODA), SIAM, 2014.

\bibitem{Kolpakov}
R.~Kolpakov, Almost optimal searching of maximal subrepetitions in a word,
  Mathematics 10~(19) (2022) 3569.

\bibitem{KolpakovKucherov}
R.~Kolpakov, G.~Kucherov, Finding maximal repetitions in a word in linear time,
  in: Proc. Foundations of Computer Science (FOCS), IEEE, 1999.

\bibitem{KolpakovKucherov2}
R.~Kolpakov, G.~Kucherov, Finding repeats with fixed gap, in: Proc. String
  Processing and Information Retrieval (SPIRE), IEEE, 2000.

\bibitem{KPPK}
R.~Kolpakov, M.~Podolskiy, M.~Posypkin, N.~Khrapov, Searching of gapped repeats
  and subrepetitions in a word, in: Proc. Combinatorial Pattern Matching (CPM),
  vol. 8486 of LNCS, Springer, 2014.

\bibitem{KPPK2}
R.~Kolpakov, M.~Podolskiy, M.~Posypkin, N.~Khrapov, Searching of gapped repeats
  and subrepetitions in a word, Journal of Discrete Algorithms 46 (2017) 1--15.

\bibitem{KosolobovLZRuns}
D.~Kosolobov, {L}empel--{Z}iv factorization may be harder than computing all
  runs, in: Proc. Symposium on Theoretical Aspects of Computer Science (STACS),
  vol.~30 of LIPIcs, Schloss Dagstuhl--Leibniz-Zentrum fuer Informatik, 2015.

\bibitem{KosolobovBitOpt}
D.~Kosolobov, Relations between greedy and bit-optimal {LZ77} encodings, in:
  Proc. Symposium on Theoretical Aspects of Computer Science (STACS), vol.~96
  of Leibniz International Proceedings in Informatics (LIPIcs), Schloss
  Dagstuhl -- Leibniz-Zentrum f{\"u}r Informatik, 2018.

\bibitem{KVNP}
D.~Kosolobov, D.~Valenzuela, G.~Navarro, S.~Puglisi, {L}empel-{Z}iv-like
  parsing in small space, Algorithmica 82~(11) (2020) 3195--3215.

\bibitem{Larsson}
N.~J. Larsson, Most recent match queries in on-line suffix trees, in: Proc.
  Combinatorial Pattern Matching (CPM), vol. 8486 of LNCS, Springer, 2014.

\bibitem{Lothaire}
M.~Lothaire, Combinatorics on words, Cambridge University Press, 1997.

\bibitem{MainLorentz}
M.~G. Main, R.~J. Lorentz, An {O}(n log n) algorithm for finding all
  repetitions in a string, Journal of Algorithms 5~(3) (1984) 422--432.

\bibitem{Nekrich}
Y.~Nekrich, New data structures for orthogonal range reporting and range minima
  queries, in: Proc. the ACM-SIAM Symposium on Discrete Algorithms (SODA),
  SIAM, 2021.

\bibitem{ParshinaPuzynina}
O.~Parshina, S.~Puzynina, Finite and infinite closed-rich words, Theoretical
  Computer Science 984 (2024) 114315.

\bibitem{Raffinot}
M.~Raffinot, On maximal repeats in strings, Information Processing Letters
  80~(3) (2001) 165--169.

\bibitem{ShibataEtAl}
H.~Shibata, H.~Umezaki, T.~Mieno, Y.~Nakashima, S.~Inenaga, Online computation
  of maximal closed substrings, in: Proc. Symposium on String Processing and
  Information Retrieval (SPIRE), vol. 17129 of LNCS, Springer, 2026.

\bibitem{SumiyoshiMienoInenaga}
W.~Sumiyoshi, T.~Mieno, S.~Inenaga, Faster and simpler online/sliding rightmost
  {L}empel-{Z}iv factorizations, in: Proc. String Processing and Information
  Retrieval (SPIRE), vol. 14899 of LNCS, Springer, 2024.

\bibitem{vanEmdeBoas}
P.~van Emde~Boas, Preserving order in a forest in less than logarithmic time,
  in: Proc. Symposium on Foundations of Computer Science (SFCS), IEEE, 1975.

\bibitem{Zhou}
G.~Zhou, Two-dimensional range successor in optimal time and almost linear
  space, Information Processing Letters 116~(2) (2016) 171--174.

\bibitem{LZ77}
J.~Ziv, A.~Lempel, A universal algorithm for sequential data compression, IEEE
  Transactions on Information Theory 23~(3) (1977) 337--343.

\end{thebibliography}

\end{document}